\documentclass[11pt]{article}

\usepackage[margin = 1in]{geometry}
\usepackage{graphicx}              
\usepackage{amsmath}               
\usepackage{amsfonts}              
\usepackage{amsthm}                
\usepackage{amssymb}
\usepackage{mathrsfs}
\usepackage{url}
\usepackage{color}
\usepackage{hyperref}
\usepackage{algpseudocode}
\usepackage{algorithm}
\usepackage{enumitem}
\usepackage{makecell}
\usepackage{booktabs}
\usepackage{tikz}
\usepackage{authblk}
\usetikzlibrary{positioning, arrows}
\usetikzlibrary{arrows.meta}

\hypersetup{
	unicode = true,
	colorlinks = true,
	citecolor = blue,
	filecolor = blue,
	linkcolor = blue,
	urlcolor = blue,
	pdfstartview = {FitH},
}

\theoremstyle{plain}
\newtheorem{theorem}{Theorem}
\newtheorem{lemma}[theorem]{Lemma}

\newtheorem{proposition}[theorem]{Proposition}
\theoremstyle{definition}

\newtheorem{conjecture}[theorem]{Conjecture}

\newtheorem*{remark}{Remark}

\algrenewcommand{\Require}{\item[\textbf{Input:}]}
\algrenewcommand{\Ensure}{\item[\textbf{Output:}]}

\newcommand{\ang}[1]{\langle#1\rangle}
\newcommand{\abs}[1]{\left\vert#1\right\vert}

\newcommand{\tildO}{\tilde{O}}

\newcommand{\Romnum}[1]{\uppercase\expandafter{\romannumeral #1}}

\DeclareMathOperator{\ringofend}{End} 

\DeclareMathOperator{\loglog}{loglog}

\DeclareMathOperator{\Hom}{Hom}

\def\Z{\ensuremath{\mathbb{Z}}}
\def\F{\ensuremath{\mathbb{F}}}
\def\P{\ensuremath{\mathbb{P}}}
\def\MM{\ensuremath{\mathsf{M}}}

\title{On Division Polynomial PIT and Supersingularity}

\author[1]{Javad Doliskani}
\affil[1]{\small Institute for Quantum Computing, University of Waterloo}

\date{}

\begin{document}

\maketitle

\begin{abstract}
	\noindent For an elliptic curve $E$ over a finite field $\F_q$, where $q$ is a prime power, we 
	propose new algorithms for testing the supersingularity of $E$. Our algorithms are based on the 
	Polynomial Identity Testing (PIT) problem for the $p$-th division polynomial of $E$. In 
	particular, an efficient algorithm using points of high order on $E$ is given.
	
	\vspace*{7mm}
\end{abstract}


\section{Introduction}
\label{sec:intro}

Recent cryptographic initiatives based on supersingular elliptic curves have attracted considerable 
attention \cite{Jao2011, Jao2014, Charles2009}. In particular, the underlying security assumption 
in such constructions is the hardness of computing isogenies between two curves. Various attempts 
on solving the \textit{Supersingular Isogeny Problem} has led to a more extensive study of 
supersingular elliptic curves. One of the interesting algorithmic questions about an elliptic curve 
is to efficiently decide whether it is ordinary or supersingular. In this paper, we propose 
efficient solutions for this question.

An elliptic curve $E / \F_q$ is supersingular if its Hasse invariant is zero, and ordinary 
otherwise. The Hasse invariant is said to be zero if the map $\pi^*: H^1(E, \mathscr{O}_E) 
\rightarrow H^1(E, \mathscr{O}_E)$ induced by the Frobenius $\pi: E \rightarrow E$ is zero. Other 
equivalent definitions of supersingularity can be derived from this one. For example, given the 
short Weierstrass equation $y^2 = x^3 + ax + b$ of $E$, the Hasse invariant is zero if and only if 
the coefficient of $x^{p - 1}$ in $(x^3 + ax + b)^{(p - 1) / 2}$ is zero. One of the more usual 
definitions of supersingularity is based on the $p$-torsion subgroup $E[p]$ of $E$. If $E[p] = 0$ 
(equivalently $E[p^n] = 0$) then $E$ is supersingular. See \cite{Husemoeller1987} for an 
introduction on elliptic curves and other definitions for supersingulariy. 

It can be shown that the $j$-invariants of all supersingular elliptic curves $E / \overline{\F}_p$ 
reside in the quadratic extension $\F_{p^2}$. Since isomorphism classes of elliptic curves are 
uniquely determined by their $j$-invariant, it follows that every supersigular elliptic curve can 
be defined over $\F_{p^2}$. So for the rest of this paper we assume $\F_q = \F_{p^2}$. Note that 
an elliptic curve $E$ with $j(E) = 0$ (resp. $j(E) = 1728$) is supersingular if and only if $p = 2 
\bmod 3$ (resp. $p = 3 \bmod 4$) \cite[Theorem 4.1.c]{silverman2009arithmetic}. Since there is only 
one supersingular curve for $p = 2$ and $p = 3$, which has $j$-invariant $j = 0$, we assume that $p 
> 3$. 

Denote by $\MM(n)$ the cost of multiplying two polynomials of degree $n$ over $\F_q$ where the unit 
cost is operations in $\F_q$. We can take $\MM(n) = O(n\log n \loglog n)$ \cite{vzGG}. We also 
assume that two $n$-bit integers can be multiplied using $O(n\log n \loglog n)$ bit operations 
\cite{vzGG}. Since $\F_q / \F_p$ is only a quadratic extension, multiplication in $\F_q$ is 
performed using a few multiplications in $\F_p$. Therefore, polynomial multiplication over $\F_q$ 
has the same asymptotic cost as that over $\F_p$, and we shall always count the number of 
operations in $\F_p$.

Given $E / \F_q$, an obvious algorithm to determine whether $E$ is ordinary or supersingular is to 
do the exponentiation above and check the coefficient of $x^{p - 1}$. The exponentiations algorithm 
then takes $O(\MM(p))$ operations in $\F_q$. A slightly better complexity $O(p)$ can be obtained by 
using the Deuring polynomial $H_p(x)$ \cite[Theorem 4.1.b]{silverman2009arithmetic}. These 
algorithms are exponential in $\log p$, and we wish to have algorithms with complexity at most 
polynomial in $\log p$.

The Frobenius endomorphism $\pi$ of $E$ satisfies $\pi^2 - t\pi + p^2 = 0$ in the endomorphism ring 
$\ringofend(E)$. The number $t$ is called the \textit{trace} of the Frobenius. One can show that 
$E$ is supersingular if and only if $t \equiv 0 \pmod{p}$. An immediate algorithm that comes to 
mind is then to compute $t$ by counting the number of points on $E$, and test the above congruence. 
An efficient deterministic point counting algorithm is due to \cite{schoof85}, which takes 
$\tildO(\log^5 p)$ bit operations. The notation $\tildO$ means ignoring logarithmic factors in the 
main complexity parameter. A Las Vegas variant of Schoof's algorithm, called SEA, was proposed by 
Elkies and Atkin. Under some heuristic assumption, the SEA algorithm runs in an expected 
$\tildO(\log^4p)$ bit operations.  

If one is content with a high-probability output then a simple \textit{Monte Carlo} test can be 
performed as follows. The curve $E$ is supersingular if and only if $E(\F_q) \cong (\Z/(p \pm 
1)\Z)^2$. Therefore, we can pick a random point $P \in E(\F_q)$ and check whether $(p \pm 1)P = 0$. 
If so, then $E$ is supersingular with probability at least $(p - 1) / p$ \cite[Prop. 
2]{sutherland2012}, which is very close to 1 when $p$ is large. The cost of this test is 
$\tildO(\log p)$ operations in $\F_p$ or $\tildO(\log^2 p)$ bit operations.

The best known deterministic algorithm, up to the author's knowledge, is due to 
\cite{sutherland2012}. The algorithm is based on traversing the isogeny graph of $E$, and it runs 
in $O(\log^3p \loglog^2p)$ bit operations. In fact, the algorithm exploits an interesting 
structural difference between isogeny graphs of ordinary and supersingular curves. For more details 
on the algorithm and on isogeny graphs we refer the reader to \cite{sutherland2012, kohel1996}. In 
this paper, we propose an algorithm that can efficiently check the supersingularity of $E$ using 
points of high order on $E$. Our main result can be summarized as follows:
\begin{theorem}
	\label{theo:main}
	Given an elliptic curve $E/\F_q$, there exists a Las Vegas algorithm that can efficiently
	decide whether $E$ is supersingular or ordinary. On input an ordinary curve, the algorithm runs 
	in an expected $\tildO(\log p)$ operations in $\F_p$. On input a supersingular curve, under the 
	generalized Riemann hypothesis, the algorithm runs in an expected $\tildO(\log^2 p)$ operations 
	in $\F_p$.
\end{theorem}
Note that, ignoring the logarithmic factors, our result does not asymptotically improve on the 
state-of-the-art algorithm \cite{sutherland2012}. However, as experiments in Section 
\ref{sec:exper} show, our algorithm attains better runtimes in practice. There is a conjecture of 
Poonen, discussed in Section \ref{sec:pit-h-order}, about orders of points in subvarieties of 
semiabelian varieties. If the conjecture holds, our algorithm always runs in an expected 
$\tildO(\log p)$ operations in $\F_p$. This is a significant improvement on the algorithm of 
\cite{sutherland2012} both in theory and in practice.

\paragraph{Polynomial identity testing (PIT)}
Given a field $K$, an arithmetic circuit with $n$ variables over $K$ is a directed acyclic graph 
with the leaves considered as input variables and the root as output. The operations on the input, 
which are implemented by the internal nodes (gates), consist of only addition and multiplication in 
$K$. Here, the edges act as wires. Therefore, a circuit implements a polynomial function in $K[x_1, 
\dots, x_n]$. The size of a circuit $C$ is defined as the number of gates in $C$. The polynomial 
identity testing can be formally stated as: 

\vspace*{2mm}

\textit{PIT Problem:} Let $f \in K[x_1, \dots, x_n]$ be a polynomial given by the arithmetic 
circuit $C$. Find a deterministic algorithm with complexity $poly(size(C))$ operations in $K$ that 
tests if $f$ is identically zero.

\vspace*{2mm}

The above is equivalent to testing $f_1 = f_2$ for two given polynomials $f_1, f_2 \in K[x_1, 
\dots, x_n]$. A very efficient probabilistic algorithm is derived from the following theorem 
\cite{schwartz1980,zippel1979}.
\begin{lemma}[Schwartz-Zippel]
	\label{lemma:Schwartz-Zippel}
	Let $f \in K[x_1, \dots, x_n]$ be of degree $d \ge 0$. Let $S \subseteq K$ be a finite subset, 
	and let $\mathbf{s} \in S^n$ be a point with coordinates chosen independently and uniformly at 
	random. Then $\Pr[f(\mathbf{s}) = 0] \le d / \abs{S}$.
\end{lemma}
If the degree $d$ is small compared to $\abs{S}$ then, by the theorem, evaluating $f$ at a random 
point tells if $f = 0$ with high probability. We will use this theorem in Section 
\ref{sec:pit-for-div}. For a survey on polynomial identity testing see \cite{saxena2009}.


\section{Division polynomials}
\label{sec:divpoly}

The $n$-th division polynomial $\psi_n$ of $E$ is an element of the function field $K(E)$ of $E$ 
with divisor $(\psi_n) = [n]^*\infty - n^2\infty$. The map $[n]^*$ is induced by the 
multiplication-by-$n$ endomorphism on the divisor class group $\text{Div}(E)$. Division polynomials 
can be defined using recursive relations as follows. Given $E: y^2 = x^3 + ax + b$ we have
\begin{equation}
\label{equ:divpoly1}
	\begin{array}{rll}
		\psi_0 & = & 0 \\
		\psi_1 & = & 1 \\
		\psi_2 & = & 2y \\
		\psi_3 & = & 3x^4 + 6ax^2 + 12bx - a^2 \\
		\psi_4 & = & 4y(x^6 + 5ax^4 + 20bx^3 - 5a^2x^2 - 4abx - 8b^2 - a^3) \\
		\psi_{2m + 1} & = & \psi_{m + 2}\psi_m^3 - \psi_{m - 1}\psi_{m + 1}^3 \quad \text{for } m 
		\ge 2 \\
		\psi_{2m} & = & (2y)^{-1}(\psi_{m + 2}\psi_{m - 1}^2 - \psi_{m - 2}\psi_{m + 1}^2)\psi_m 
		\quad \text{for } m \ge 3.
	\end{array}
\end{equation}
We also define the following polynomials which will be used in the subsequent sections.
\begin{equation}
\label{equ:divpoly-extra}
	\begin{array}{rll}
		\phi_m & = & x\psi_m^2 - \psi_{m + 1}\psi_{m - 1} \\
		\omega_m & = & (4y)^{-1}(\psi_{m + 2}\psi_{m - 1}^2 - \psi_{m - 2}\psi_{m + 1}^2)
	\end{array}
\end{equation}
It follows from the definition that a point $P \in E(\overline{\F}_p)$ is an $n$-torsion if and 
only if $\psi_n(P) = 0$. In other words, the division polynomial $\psi_n$ exactly encodes the whole 
$n$-torsion $E[n]$. Division polynomials play an important role in the theory of elliptic 
curves. They are also heavily used in implementations of the point counting algorithm of 
\cite{schoof85}.

\subsection{Computing division polynomials}

In this subsection, we briefly review an algorithm that efficiently computes the division 
polynomials. To compute $\psi_n$, the idea is to simply use the recursive relations 
\eqref{equ:divpoly1} to achieve a double-and-add scheme on the subscript $n$. First, we should note 
that it is possible to characterize non-$2$-torsion points with univariate versions of the 
$\psi_n$. This makes computations much easier. Define
\[
f_m = 
\begin{cases}
	\psi_m & m \text{ odd} \\
	\psi_m / \psi_2 & m \text{ even}.
\end{cases}
\]
Then for non-$2$-torsion $P \in E$ we have $P \in E[n]$ if and only if $f_n(P) = 0$. Let $F = 
\psi_2^2 = 4(x^3 + ax + b)$. Then the following relations for the $f_m$ can be derived from 
\eqref{equ:divpoly1}.
\begin{equation}
\label{equ:divpoly2}
	\begin{array}{rll}
		f_0 & = & 0 \\
		f_1 & = & 1 \\
		f_2 & = & 1 \\
		f_3 & = & \psi_3 \\
		f_4 & = & \psi_4 / \psi_2 \\
		f_{2m + 1} & = & 
		\begin{cases}
			f_{m + 2}f_m^3 - F^2f_{m - 1}f_{m + 1}^3 & m \text{ odd, } m \ge 3 \\
			F^2f_{m + 2}f_m^3 - f_{m - 1}f_{m + 1}^3 & m \text{ even, } m \ge 2 \\
		\end{cases} \\
		f_{2m} & = & (f_{m + 2}f_{m - 1}^2 - f_{m - 2}f_{m + 1}^2)f_m \quad \text{for } m \ge 3.
	\end{array}
\end{equation}
From the indices involved in the above relations it is immediate that given $f_{i - 3}, \dots, f_{i 
+ 5}$, one can compute the polynomials $f_{2i - 3}, \dots, f_{2i + 5}$, or the polynomials $f_{2(i 
+ 1) - 3}, \dots, f_{2(i + 1) + 5}$. We can save some multiplications by introducing $S_i = f_{i - 
1}f_{i + 1}$, $T_i = f_i^2$ and rewriting \eqref{equ:divpoly2} as
\begin{equation}
\label{equ:divpoly3}
\begin{array}{rll}
	f_{2m + 1} & = & 
	\begin{cases}
		T_mS_{m + 1} - F^2T_{m + 1}S_m & m \text{ odd, } m \ge 3 \\
		F^2T_mS_{m + 1} - T_{m + 1}S_m & m \text{ even, } m \ge 2 \\
	\end{cases} \\
	f_{2m} & = & T_{m - 1}S_{m + 1} - T_{m + 1}S_{m - 1} \quad \text{for } m \ge 3.
\end{array}
\end{equation}
We shall only need modular computation of division polynomials in this paper. Therefore, Algorithm 
\ref{alg:divpoly-comp} performs computations mod $f$ for a given polynomial $f \in \F_q[x]$.

\begin{algorithm}[H]
	\caption{Division polynomial computation}
	\label{alg:divpoly-comp}
	\begin{algorithmic}[1]
		\Require Integer $m \ge 1$, and polynomial $f(x) \in \F_q[x]$ of degree $n$
		\Ensure The division polynomials $f_{m - 3}, \dots, f_{m + 5} \pmod{f}$
		\State Let $b_kb_{k - 1} \cdots b_0$ be the binary representation of $m$
		\State Set $r, s$ as follows \\
		$r = 3, s = 2$~ if~ $b_kb_{k - 1} = 11$ \\
		$r = 4, s = 3$~ if~ $b_kb_{k - 1}b_{k - 2} = 100$ \\
		$r = 5, s = 3$~ if~ $b_kb_{k - 1}b_{k - 2} = 101$
		\State Set $j = r$, and compute $f_{j - 3}, \dots, f_{j + 5} \pmod{f}$
		\For {$i = k - s$ down to $0$}\label{step:divpoly-loop} 
			\State Compute $f_{2j + b_i - 3}, \dots, f_{2j + b_i + 5} \pmod{f}$ from $f_{j - 3}, 
			\dots, f_{j + 5}$ using relations \eqref{equ:divpoly3}
			\State $j \leftarrow 2j + b_i$
		\EndFor
		\State \Return $f_{m - 3}, \dots, f_{m + 5}$
	\end{algorithmic}
\end{algorithm}

The correctness of the algorithm follows from the previous remarks. The runtime is dominated by 
the for-loop at Step \ref{step:divpoly-loop}. The number of iterations is $O(\log m)$, and at each 
iteration a few polynomial multiplications of degree $n$ is done at the cost of $O(\MM(n))$ 
operations in $\F_p$. Therefore, the total runtime is $O(\MM(n)\log m)$ operations in $\F_p$.


\section{PIT for the $p$-th division polynomial}
\label{sec:pit-for-div}

Given an elliptic curve $E/\F_q$, we have $E[p] \cong 0$ (resp. $E[p] \cong \Z / p\Z$) when $E$ is 
supersingular (resp. ordinary). Therefore, the $p$-th division polynomial $\psi_p$ of $E$ is a 
constant $c \in \F_q$ when $E$ is supersingular, and non-constant otherwise. This means solving the 
polynomial identity testing problem $\psi_p = c$ will give an answer to the question of whether $E$ 
is ordinary or supersingular. 

First, we investigate the shape of $\psi_p$ in both cases. For the general polynomial $\psi_n$ one 
has
\[
\psi_n =
\begin{cases}
	y(nx^{(n^2 - 4) / 2} + \cdots) & n \text{ even} \\
	nx^{(n^2 - 1) / 2} + \cdots & n \text{ odd},
\end{cases}
\]
see \cite{washington2008} for more details. It follows that for $\psi_p$, which is a univariate 
polynomial, $\deg \psi_p < (p^2 - 1) / 2$. In fact, for ordinary curves we have
\begin{lemma}
	\label{lem:ord_divpoly}
	Let $E/\F_q$ be an ordinary elliptic curve and let $r$ be the order of trace $t$ of the 
	Frobenius in $(\Z / p\Z)^*$. Then $\deg \psi_p = p(p - 1) / 2$. Also
	\[ \psi_p = f_1^pf_2^p \cdots f_{(p - 1) / 2r}^p \]
	is the factorization of $\psi_p$ over $\F_q$ where all $f_i$ are of the same degree $r$.
\end{lemma}
\begin{proof}
	We have $\pi \circ \hat{\pi} = [p]$ where $\hat{\pi}$ is the dual\footnote{The dual of an 
	isogeny $\phi: E_1 \rightarrow E_2$ of degree $m$ is a unique isogeny $\hat{\phi}: E_2 
	\rightarrow E_1$ such that $\phi \circ \hat{\phi} = [m]$, see  
	\cite[III.6]{silverman2009arithmetic}.} of $\pi$. Let $\hat{\pi}(x, y) = (F, G)$ where $F, G$ 
	are rational functions in $x, y$. Then
	\begin{equation}
	\label{equ:mult-p}
		(F^p, G^p) = \pi(F, G) = [p](x, y) = \left( \frac{\phi_p}{\psi_p^2}, 		
		\frac{\omega_p}{\psi_p^3} \right).
	\end{equation}
	The last equality is the formula for multiplication by $p$, where $\phi_p$ and $\omega_p$ are
	defined using \eqref{equ:divpoly-extra}. It follows that $\psi_p$ is a $p$-th power. Since 
	$\psi_p$ has roots exactly the abscissas of the nonzero $p$-torsion points it must have degree 
	$p(p - 1) / 2$, which proves the first part.
	
	For the second part, note that the action of the Frobenius on $E[p]$ is just multiplication by 
	the trace $t$. By the first part $\psi_p = \tilde{\psi}_p^p$ where $\tilde{\psi}_p$ splits into 
	factors of degree $r$ over $\F_q$.
\end{proof}
For supersingular elliptic curves we have
\begin{lemma}
	\label{lem:ss_divpoly}
	Let $E/\F_q$ be a supersingular elliptic curve with $j(E) \ne 0, 1728$. Then $\psi_p = \pm 1$.
\end{lemma}
\begin{proof}
	When $E$ is supersingular we have $t = \pm 2p$ so that the characteristic polynomial of the 
	Frobenius factorizes as $(X \pm p)^2$. Therefore, $\pi(x, y) = \pm [p](x, y)$. Comparing the 
	first coordinates and considering the multiplication-by-$p$ formula \eqref{equ:mult-p} gives 
	$x^{p^2} = \phi_p(x) / \psi_p(x)^2$. Since $\phi_p(x) = x^{p^2} + \cdots$	
	(see~\cite{washington2008}), it follows that $\psi_p(x) = \pm 1$.
\end{proof}
It is not hard to distinguish the cases $\psi_p = 1$ and $\psi_p = -1$ in Lemma 
\ref{lem:ss_divpoly}, but we are not concerned with that. In fact, it is easily done in practice 
by computing $\psi_p(0)$. Without loss of generality, we only consider the case $\psi_p(x) = 1$. 
The remainder of this section is devoted to two algorithms derived from the above results. An 
efficient algorithm based on points of high order is discussed in Section \ref{sec:pit-h-order}.

\paragraph{A probabilistic algorithm.}
Lemmas \ref{lem:ord_divpoly}, \ref{lem:ss_divpoly} and Lemma \ref{lemma:Schwartz-Zippel} give a 
probabilistic algorithm for the PIT $\psi_p(x) = 1$: select a random element $a \in S = \F_q$ and 
compute $\psi_p(a)$. If $\psi_p(a) = 1$ then the algorithm outputs ``supersingular'', otherwise it 
outputs ``ordinary''. If $E$ is ordinary, then the output is ``supersingular'' with probability 
\[ P[\psi_p(a) = 1] = P[\tilde{\psi}_p(a) = 1] \le (p - 1) / 2p^2 < 1 / 2p, \] 
where $\psi_p = \tilde{\psi}_p^p$ as in Lemma \ref{lem:ord_divpoly}. The runtime of this algorithm 
is $\tildO(\log p)$ operations in $\F_p$, which is the same as the Monte Carlo algorithm given in 
Section \ref{sec:intro}.

\paragraph{A Schoof-like algorithm}
To see whether $E/\F_q$ is supersingular one can check either of the identities $\psi_p(x) = 1$ or 
$\pi = \pm [p]$. Therefore, solving the former PIT is equivalent to testing the latter identity. 
This can be done using arithmetic modulo division polynomials as we show in the following. See 
\cite[Section 2.2]{sutherland2012} for a similar algorithm.

Let $\mathcal{S} = \{2, 3, \dots, \ell\}$ be a set of primes such that $\prod_{r \in \mathcal{S}} r 
\ge 4p$. For this to be true we only need $\ell = O(\log p)$. We know that $E$ is supersingular if 
and only if $t = \pm 2p$. Therefore, it follows from the Chinese Remainder Theorem that $E$ is 
supersingular if and only if $t \equiv \pm 2p \pmod{r}$ for all $r \in \mathcal{S}$. Let $P \in E$ 
be a point of prime order $r$. Evaluating the characteristic polynomial of $\pi$ at $P$ gives
\[ \pi^2(P) - t\pi(P) + [p^2]P = 0. \]
If also $\pi(P) = \pm [p]P$ then $t \equiv \pm 2p \pmod{r}$. So, by the above, $E$ is supersingular 
if and only if $\pi(P) = \pm [p]P$ for a point $P$ of order $r$ for all $r \in \mathcal{S}$.

The last condition can be checked using division polynomials. More precisely, it is equivalent to 
checking the identities
\begin{equation}
\label{equ:schoof}
	x^{p^2} = x - \frac{\phi_s(x)}{\psi_s^2(x)} ~ \bmod \psi_r(x), \quad \text{for all } r \in 
	\mathcal{S},
\end{equation}
where $s = p \bmod r$. The polynomials $\{ \psi_i \}_{i \le \abs{\mathcal{S}}}$ can be computed in 
negligible time. Since the degree of $\phi_r(x)$ is $r^2 = O(\log^2 p)$, computing $x^{p^2} \bmod 
\psi_r$ takes $\tildO(\log^3 p)$ operations in $\F_p$. So the total cost of these checks is 
$\tildO(\log^4 p)$ operations in $\F_p$.


\section{PIT using points of high order}
\label{sec:pit-h-order}

From the identity $\psi_p(x) = 1$ for a supersingular curve $E/\F_q$ and the multiplication-by-$p$ 
formula \eqref{equ:mult-p} we get the stronger condition
\begin{equation}
\label{equ:ss-strong}
	\arraycolsep=1.4pt
	\begin{array}{l}
		\psi_p(x) = 1 \\
		\psi_{p - 1}\psi_{p + 1}(x) = x - x^{p^2}
	\end{array}
\end{equation}
for supersingularity. Note that in this case the univariate polynomial $\psi_{p - 1}\psi_{p + 
1}(x)$ splits completely over $\F_q$. If these identities are not satisfied when evaluated at any 
point $P \in E$ then $E$ is ordinary. Otherwise, not much can be said about $E$ unless the point 
$P$ is chosen more carefully. More precisely, we have
\begin{proposition}
	Let $P \in E$ be a point of order $r$, with $(r, p) = 1$, that satisfies \eqref{equ:ss-strong}. 
	Then $\psi_p(kP) = 1$ for all odd $1 \le k \le r - 1$. If moreover $\psi_p(2P) = 1$ then 
	$\psi_p(kP) = 1$ for all $1 \le k \le r - 1$.
\end{proposition}
\begin{proof}
	For all positive integers $m, n$ we have
	\begin{equation}
	\label{equ:divpoly-comp}
		\psi_{mn} = (\psi_m \circ [n])\psi_n^{m^2}.
	\end{equation}
	This follows from comparing divisors on both sides. When $k < r$ is odd, $\psi_k$ is a 
	univariate polynomial and so $\psi_k(P)$ only involves the $x$-coordinate of $P$. Since $x_{pP} 
	= \phi_p / \psi_p^2$ where $\phi_p$ is defined by \eqref{equ:divpoly-extra}, it follows that if 
	$P \in E$ satisfies \eqref{equ:ss-strong} then $x_{pP} = x_{\pi P}$ and hence 
	\begin{equation}
	\label{equ:frob-p}
		\psi_k(P)^{p^2} = \psi_k(\pi(P)) = \psi_k(pP).
	\end{equation}
	Therefore, for $(m, n) = (p, k)$ and $(m, n) = (k, p)$ identity \eqref{equ:divpoly-comp} gives
	\[
	\arraycolsep=1.4pt
	\begin{array}{rll}
		\psi_{kp}(P) & = & \psi_p(P)^{k^2} \psi_k(pP) \\
		& = & \psi_k(pP), \vspace*{2mm} \\		 
		\psi_{kp}(P) & = & \psi_k(P)^{p^2} \psi_p(kP) \\
		& = & \psi_k(pP) \psi_p(kP), \\
	\end{array}
	\]
	where $\psi_k(pP) \ne 0$ since $P$ is of order $r > k$. Comparing the above identities proves 
	the first part. For the second claim we show that \eqref{equ:frob-p} holds for any $k$, and the 
	rest of the proof is the same. Using \eqref{equ:divpoly-comp} with $(m, n) = (2, p)$ and $(m, 
	n) = (p, 2)$ we get $2y_{\pi P} = \psi_2(\pi(P)) = \psi_2(pP) = 2y_{pP}$ which implies $\pi(P) 
	= pP$, and hence $\psi_k(\pi(P)) = \psi_k(pP)$ for any $k$.
\end{proof}
The above result enables us to check the PIT \eqref{equ:ss-strong} using a high order point as 
follows. Let $P \in E$ be a point of order $r > 2p + 2$. If $P$ satisfies \eqref{equ:ss-strong} 
then $\tilde{\psi}_p(kP)^p = \psi_p(kP) = 1$ hence $\tilde{\psi}_p(kP) = 1$ for at least $p + 1$ 
values of $k$. This means $\tilde{\psi}_p(x_{kP}) = 1$ for at least $(p + 1) / 2$ distinct 
abscissas of the points $kP$. But the univariate polynomial $\tilde{\psi}_p$ has degree at most $(p 
- 1) / 2$ so it is uniquely determined by $(p + 1) / 2$ pairs $(a, \tilde{\psi}_p(a))$. It follows 
that $\tilde{\psi}_p^p = 1$, and hence $\psi_p = 1$. It only remains to efficiently find a point of 
high order on $E$. For this, we adapt the approach of Voloch \cite{voloch2007,voloch2010}.

Let $\mathbb{G}_m$ be the multiplicative group over $\F_q$. In \cite{voloch2010}, elements of high 
order in $\F_q^*$ are obtained using points on a curve contained in the fibered product $E \times 
\mathbb{G}_m$. We can use the same technique to obtain points of high order on $E$. The proofs 
remain essentially the same except for some parts which we explain in the following. We fix an 
embedding $\mathbb{G}_m \hookrightarrow \P^1$ and assume all curves are projective.

Let $X$ be an absolutely irreducible curve contained in $E \times \P^1$. Also assume that $X$ 
has non-constant projections to both factors and denote by $D$ the degree of $X \rightarrow \P^1$. 
Consider the pullback diagram

\begin{center}
	\begin{tikzpicture}[node distance = 1.3cm, 
		>={Computer Modern Rightarrow[length = 5, width = 4]}]
		\node (Xn) {$X_n$};
		\node[right = of Xn, inner sep = 5pt] (X) {$X$};
		\node[below = of Xn] (Gm1) {$\P^1$};
		\node[below = of X] (Gm2) {$\P^1$};
		\node[right = of X] (E) {$E$};
		\draw[->] (Xn) edge (X);
		\draw[->] (Xn) edge (Gm1);
		\draw[->] (X) edge (Gm2);
		\draw[->] (Gm1) edge node[above, midway] {\scriptsize $\mu_n$} (Gm2);
		\draw[->] (X) edge (E);
	\end{tikzpicture}
\end{center}

where $X_n$ is the fibered product $X \times \P^1$. The bottom morphism corresponds to the field 
extension $\F_q(u)/\F_q(t)$ with $u^n = t$. If $n$ is coprime to $Dp$ then the morphism $X_n 
\rightarrow X$ is separable of degree $n$, and $X_n$ is also absolutely irreducible. The morphism 
$X_n \rightarrow E$ obtained by composing the top two morphisms determines an element $y_n$ in the 
function field $K_n = K(X_n)$. Assume that all such $y_n$ are elements of some fixed algebraic 
closure of $K(X)$. Then Lemma 2.2 in \cite{voloch2010} becomes
\begin{lemma}
	\label{lem:linear-ind}
	The functions $y_n$, considered as morphisms $X_n \rightarrow E$ with $(n, Dp) = 1$, are 
	$\Z$-linearly independent.
\end{lemma}
\begin{proof}
	Given $\{y_{n_i}\}_{1 \le i \le s}$, let $L$ be the compositum of the function fields $K_{n_i}$,
	and let $C_L$ be the smooth curve with function field $L$. We have an isomorphism 
	\begin{equation}
	\label{equ:curve-hom}
		E(L) \cong \Hom_{\F_q}(C_L, E),
	\end{equation}
	where the right hand side is the group of morphisms of $k$-schemes. This isomorphism holds if 
	$(L, C_L)$ are replaced by any $(K_{n_i}, X_{n_i})$. To linearize the addition of the $y_{n_i}$ 
	on $E$ to addition of differentials in $\Omega_{C_L}$ we consider the pullbacks $\omega_{n_i} = 
	y_{n_i}^*(\omega_E)$ of the invariant differential $\omega_E$ of $E$. Since $C_L \rightarrow 
	X_{n_i}$ is separable for all $1 \le i \le s$, the induced morphisms
	\begin{equation}
	\label{equ:diff-mod}
	y_{n_i}^*: \Omega_{X_{n_i}} \longrightarrow \Omega_{C_L}
	\end{equation}
	are injective. Now it follows from \eqref{equ:curve-hom} and \eqref{equ:diff-mod} that the 
	$\omega_{n_i}$ reside in different extensions inside $\Omega_{C_L}$ so they must be 
	$\Z$-linearly independent.
\end{proof}
We need Lemma 2 from \cite{voloch2007} which we state here for convenience.
\begin{lemma}
	\label{lem:coset}
	Let $m, q \ge 2$ be fixed integers and $\epsilon > 0$ a real number. For an integer $r \ge 2$ 
	with $(r, mq) = 1$, if $d$ is the order of $q \bmod r$, then, given $N < d$, there is a coset 
	$\Gamma$ of $\ang{q} \subset (\Z / r)^*$ with
	\[ \# \{ n \mid 1 \le n \le N, (n, m) = 1, n \bmod r \in \Gamma \} \gg Nd^{1 - \epsilon} / r - 
	r^\epsilon. \]
\end{lemma}
Following the notation of \cite{voloch2010}, for a function field $L/\F_q$ we define $\deg_Lz$ to 
be the degree of the divisor of zeros of $z$ in $L$. So $\deg_{K_n}y_n \ll n$.
\begin{theorem}
	\label{theo:order-general}
	Let $X \subset E \times \mathbb{G}_m$ be an absolutely irreducible curve over $\F_q$ with 
	non-constant projections to both factors. Given $\epsilon > 0$ there exists $\delta > 0$ such 
	that if $(P, b) \in X$ satisfies
	\begin{enumerate}[leftmargin = *, labelsep = *, align = left, itemsep = -0.1cm, font = 
	\normalfont, label = (\roman*)]
		\item $d = [\F_q(b) : \F_q]$ sufficiently large,
		\item $\ang{P}$ invariant under the action of $\pi$,
		\item the order $r$ of $b$ satisfies $r < d^{2 - \epsilon}$
	\end{enumerate}
	then $P$ has order at least $\exp(\delta (\log d)^2)$.
\end{theorem}
\begin{proof}
	A slight modification of the proof of \cite[Theorem 1.1]{voloch2010} works here. Let $N = [d^{1 
	- \epsilon}]$. Then using the parameters in the conditions of the theorem we get a coset 
	$\Gamma = \gamma \ang{q}$ from Lemma \ref{lem:coset}. Let $c \in \overline{\F}_q$ be an $r$-th 
	root of unity such that $c^\gamma = b$. For an integer $n \le N$ with $(n , q) = 1$ and $n 
	\bmod r \in \Gamma$ we have, by construction, $n \equiv \gamma q^j \bmod r$ for some $j$. Let 
	$J$ be the set of $j$'s obtained for all such $n$. For simplicity, denote also by $\pi$ the 
	$\F_q$-Frobenius map on $E \times \mathbb{G}_m$. Then for $j \in J$ we have
	\[ \pi^j(P, b) = (\pi^j(P), \pi^j(b)) = (\pi^j(P), b^{q^j}) = (\pi^j(P), c^{n_j}) \]
	where $n_j$ corresponds to $j$, that is $n_j \le N$, $(n_j, q) = 1$, $n_j \bmod r \in 
	\Gamma$ and $n_j \equiv \gamma q^j \bmod r$. This means that there is a place of $K_{n_j}$ above 
	$c$ where $y_{n_j}$ takes the value $\pi^j(P)$. 
	
	For a subset $I \subset J$ define $P_I = \sum_{j \in I}\pi^j(P)$. Note that $P_I$ is in 
	$\ang{P}$ since $\ang{P}$ is invariant under $\pi$. Let $T = [\eta \log d]$ for some real 
	parameter $\eta > 0$. We show that for distinct $I \subset J$ with $\abs{I} \le T$ the points 
	$P_I$ are distinct. Assume $P_I = P_{I'}$ for $I \ne I'$. Define  
	\[ z = \sum_{j \in I}y_{n_j} - \sum_{j \in I'}y_{n_j} \]
	where the addition occurs on $E$. Let $L$ be the compositum of $\{ K_{n_j} \}_{j \in I \cup 
	I'}$. Then $z$ vanishes at a place of $L$ above $c$. But we have
	\[ \deg_L z \le \sum_{j \in I \cup I'} \deg_L y_j = \sum_{j \in I \cup I'} [L : K_{n_j}] 
	\deg_{K_{n_j}} y_{n_j} \ll TD^{2T}N \] 
	which can be made smaller than $d = [\F_q(c) : \F_q]$ for some small $\eta$ and all 
	sufficiently large $d$. This is not possible unless $z = 0$ and so the $ \{ y_{n_j} \}_{j \in I 
	\cup I'}$ are $\Z$-linearly dependent, which contradicts Lemma \ref{lem:linear-ind}. Therefore, 
	the number of distinct points $P_I$ is at least $\binom{\abs{J}}{T}$. Setting $\epsilon 
	\leftarrow \epsilon / 3$ in Lemma \ref{lem:coset} we get the same bound as in \cite{voloch2007}:
	\[ \abs{J} \gg d^{2 - \epsilon / 3} / r - r^{\epsilon / 3} \gg d^{2\epsilon / 3}, \]
	and $\binom{\abs{J}}{T} \ge (\abs{J} / T - 1)^T \gg \exp(\delta (\log d)^2)$ for some $\delta > 
	0$.
\end{proof}
The logarithmic term $\log d$ in the exponent in Theorem \ref{theo:order-general} is forced by the 
exponent $2T$ in the bound $TD^{2T}N$ obtained in the proof. For special cases of the irreducible 
curve $X$, e.g. an open subset of the graph of a morphism $f: E \rightarrow \mathbb{P}^1$, we can 
obtain much better bounds on the order of $P$.  
\begin{theorem}
	\label{theo:order-ell}
	Let $E/\F_q$ be an elliptic curve and let $f: E \rightarrow \mathbb{P}^1$ be the projection to 
	the first coordinate. Then with the assumptions of Theorem \ref{theo:order-general} and with 
	$X$ as an open subset of the graph of $f$, the point $P$ has order at least $\exp(d^\delta)$.
\end{theorem}
\begin{proof}
	Let $X$ be an open subset of the graph of a morphism $f: E \rightarrow \mathbb{P}^1$ defined by 
	projection to the first coordinate, that is $f(P) = x_P$ for any point $P$ on $E$. A point $(P, 
	x_P)$ is in $X_n$ if $(P, x_P^n)$ is in $X$. Therefore, it is implied by the commutative digram
	\begin{center}
		\begin{tikzpicture}[>={Computer Modern Rightarrow[length = 5, width = 4]}]
			\node[inner sep = 5pt] (KXn) {$K_n$};
			\node[right = 1.7cm of KXn, inner sep = 3pt] (KX) {$K(X)$};
			\node[below = of Xn] (KGm1) {$K(\P^1)$};
			\node[below = of KX] (KGm2) {$K(\P^1)$};
			\node[right = of KX] (KE) {$K(E)$};
			\draw[->] (KX) edge (KXn);
			\draw[->] (KGm1) edge (KXn);
			\draw[->] (KGm2) edge (KX);
			\draw[->] (KGm2) edge node[above, midway] {\scriptsize $\mu_n^*$} (KGm1);
			\draw[->] (KE) edge (KX);
		\end{tikzpicture}
	\end{center}
	that $y_n = (x^n, y)$. So, following the proof of Theorem \ref{theo:order-general}, a much 
	smaller bound $\deg_L z \ll TDN$ is obtained. This means we can choose a larger value of $T$, 
	say $T = [d^\eta]$ following the notation of \cite{voloch2010}. Now the calculation of Theorem 
	\ref{theo:order-general} implies that $P$ has order at least $\exp(d^\delta)$ for some suitable
	$\delta > 0$.
\end{proof}
One concludes from the above theorems and the ones in \cite{voloch2010} that given a point $P \in 
E$ and a function $f$ on $E$, under some conditions, one of $P, f(P)$ has large order in its 
respective group. Therefore, forcing $f(P)$ to be of small order in $\F_q^*$ yields $P$ of large 
order in $E$. For appropriate choices of $\epsilon$ in Theorem \ref{theo:order-ell}, one gets large 
enough $\delta$, and hence $P$ of large order, without having to choose $d$ very large. Experiments 
show, however, that the lower bound of the theorem is very far from optimal. 

From the proof of Theorem \ref{theo:order-general} we see that the point $P$ has order at least 
$\binom{\abs{J}}{T}$ where $\abs{J} \gg d^{2 - \epsilon / 3} / r - r^{\epsilon / 3}$ and $T = 
[d^\eta]$ for some $\eta > 0$. The bound $\exp(d^\delta)$ is just an approximation of the binomial 
coefficient. In practice, we could ignore this approximation and use the value of the binomial 
coefficient directly. Here, $\eta$ is a function of $\epsilon$ and can be calculated from the 
inequality $\deg_L(\sum_{j \in I}y_{n_j} - \sum_{j \in I'}y_{n_j}) \ll TDN $. From this we obtain 
the rough estimate $T \approx d^{2\epsilon} / 3$ which gives $\eta \approx 2\epsilon$. 

Let $r$ be a prime such that $p$ is a generator of $(\Z/r\Z)^*$. Then $q = p^2$ has order $d = (r - 
1) / 2$, and $d^{2 - \epsilon} > r$ for a wide range of values of $0 < \epsilon < 1$. Now the 
$r$-th cyclotomic polynomial $\Phi_r(T)$ splits into two irreducible factors of degree $d$ over 
$\F_q$. Let $g(T)$ be one of the factors so that $K = \F_q[T]/g(T)$ is a field, and let $t$ be the 
image of $T$ in $K$. A point $P \in E$ with first coordinate $t$ lies in $E(F)$ where $[F : K] \le 
2$. If $P$ satisfies \eqref{equ:ss-strong} then $\ang{P}$ is invariant under the action of $\pi$. 
Then, by the above, the order of $P$ is at least $\binom{\abs{J}}{T}$. We need to choose a suitable 
$r$ and optimize for $\epsilon$ so that
\begin{equation}
	\label{equ:binom}
	\binom{\abs{J}}{T} \gg \binom{d^{2 - \epsilon / 3} / r - r^{\epsilon / 3}}{d^{2\epsilon}} 
	\approx \binom{\lfloor r / 2 \rfloor^{1 - \epsilon / 3} - r^{\epsilon / 3}}{\lfloor r / 2 
	\rfloor^{2\epsilon}} \ge 2p + 2.
\end{equation}
It suffices to take $r \in O(\log p)$ and an appropriate $\epsilon$ to obtain the bound in 
\eqref{equ:binom}. Then $P$ will have order large enough to imply that $\psi_p = 1$. 

\begin{algorithm}[H]
	\caption{Testing supersingularity}
	\label{alg:ss-strong}
	\begin{algorithmic}[1]
		\Require An elliptic curve $E$ over $\F_q$
		\Ensure True if $E$ is supersingular, and false otherwise
		\If {$j(E) = 0$}
			\State If $p = 2 \bmod 3$ then return true, otherwise return false
		\EndIf
		\If {$j(E) = 1728$}
			\State If $p = 3 \bmod 4$ then return true, otherwise return false
		\EndIf
		\State\label{step:ss-eval}%
		Compute $\psi_p(a)$ for a random $a \in \F_q$.
		\If {$\psi(a) \ne \pm 1$}
			\State return false
		\EndIf
		\State\label{step:ss-find-par}%
		Find $r \in \tildO(\log p)$ such that $p$ generates $(\Z/r\Z)^*$ and such that 
		\eqref{equ:binom} holds
		\State\label{step:cyclo}%
		Obtain an irreducible factor $g(x)$ of the $r$-th cyclotomic polynomial $\Phi_r(x)$ over 
		$\F_q$
		\State\label{step:ss-psi}%
		Compute $x^{p^2}, f_p, f_{p - 1}, f_{p + 1} \bmod g(x)$ using Algorithm 
		\ref{alg:divpoly-comp}
		\If {\eqref{equ:ss-strong} holds}
			\State return true
		\Else
			\State return false
		\EndIf
	\end{algorithmic}
\end{algorithm}

To analyze Algorithm \ref{alg:ss-strong} we need the following theorem \cite{hooley1967artin, 
matthews1976generalisation, finch2003mathematical}.
\begin{theorem}
	\label{theo:Artin}
	Let $S(p, x)$ be the number of primes $r \le x$ such that $p \bmod r$ is a generator of 
	$(\Z/r\Z)^*$. Assuming the generalized Riemann hypothesis, we have 
	\[ S(p, x) \approx C(p)\frac{x}{\log(x)}\]
	where $C(p)$ is a constant depending on $p$.
\end{theorem}
The constant $C(p)$ in Theorem \ref{theo:Artin} can be explicitly written as $C(p) = (1 + 1 / (p^2 
- p - 1))C_{\text{Artin}}$ where $C_{\text{Artin}} = 0.3739558136\dots$ is called Artin's 
constant. The theorem simply states that the density of primes for which a given prime is a 
primitive root is roughly $C(p)$. Theorem \ref{theo:main} follows from the following proposition.
\begin{proposition}
	\label{prop:ss-strong}
	Algorithm \ref{alg:ss-strong} is correct, and on input an ordinary curve, runs in an expected 
	$\tildO(\log p)$ operations in $\F_p$. On input a supersingular curve, assuming the generalized 
	Riemann hypothesis, the algorithm runs in an expected $\tildO(\log^2 p)$ operations in $\F_p$.
\end{proposition}
\begin{proof}
	Step \ref{step:ss-eval} is done using $\tildO(\log p)$ operations in $\F_p$ using Algorithm 
	\ref{alg:divpoly-comp}. Since most curves $E/\F_q$ are ordinary and they almost always fail to 
	satisfy the condition $\psi(a) \ne \pm 1$, this will be the average-case complexity of the 
	algorithm. According to Theorem \ref{theo:Artin}, the integer $r$ in Step  
	\ref{step:ss-find-par} always exists and it can be computed in negligible time. 
	
	The cyclotomic polynomial $\Phi_r(x)$ in Step \ref{step:cyclo} can be factored using 
	$\tildO(\log^2p)$ operations in $\F_p$ \cite{shoup1994}. Step \ref{step:ss-psi} is performed at 
	the cost of $\tildO(\log^2p)$ operations in $\F_p$ using Algorithm \ref{alg:divpoly-comp}. 
	Therefore, the worst-case complexity of Algorithm \ref{alg:ss-strong} is $\tildO(\log^2 p)$ 
	operations in $\F_p$.
\end{proof}

\begin{remark}
	Our experiments have been better estimated by the following stronger lower bound conjectured by 
	Poonen\footnote{See \cite{voloch2007}.}.
	\begin{conjecture}[Poonen]
		\label{conj:poonen}
		Let $X/\F_q$ be a semiabelian variety and let $Y \subset X$ be a closed subvariety. Let $Z$ 
		be the union of all translates of positive-dimensional semiabelian varieties $X' / 
		\overline{\F}_q$ contained in $X$. Then there is a constant $c > 0$ such that for every 
		nonzero $x \in (X - Z)(\overline{\F}_q)$, $x$ has order at least $(\#\F_q(x))^c$.
	\end{conjecture}
	Chang et al.\cite{chang2014} have obtained strong results indirectly confirming the above 
	conjecture for general varieties. In our context, this implies that if $P \in E(\F_q)$ does 
	not lie in any subfield, and $f$ is a non-constant function on $E$, then either $P$ or $f(P)$ 
	has order at least $(\#\F_q(P))^c$. 
	
	The point $P \in E(K)$ used in Algorithm \ref{alg:ss-strong} also satisfies the hypothesis of 
	Conjecture \ref{conj:poonen}. If the conjecture holds, then $P$ has order $\ge q^{rc} = 
	p^{2rc}$. This means we only need to take $r \approx 1 / 2c$. In this case, Algorithm 
	\ref{alg:ss-strong} always runs in $\tildO(\log p)$ operations in $\F_p$. Although, according 
	to our experiments, values of $r$ obtained this way are small, in theory we do not know of any 
	explicit bounds on $c$. This does not allow us to use Conjecture \ref{conj:poonen} for testing 
	the supersingularity of $E$.
\end{remark}


\section{Experiments}
\label{sec:exper}

We have implemented Algorithm \ref{alg:ss-strong} of this paper and Algorithm 2 of 
\cite{sutherland2012} both in C++. The arithmetic in $\F_{p^2}[x]$ is done using the NTL library 
\cite{shoup2001ntl}. The timings are obtained on a single core of an AMD FX(tm)-8120 at 1.4GHz on a 
Linux machine. Table \ref{table:exper} compares the runtimes for different sizes of the base field 
$\F_p$. The first column is the size of a randomly selected prime $p$ in bits. 

For each prime, we have generated 10 random ordinary and 10 random supersingular curves using Sage 
\cite{stein2008sage}. Since most of the elliptic curves $E/\F_q$ are ordinary, generating random 
ordinary curves amounts to simply choosing random coefficients $a, b$ for the Weierstrass equation. 
Generating supersingular curves can be efficiently done using the Complex Multiplication method in 
\cite{broker2009}. As pointed out in \cite{sutherland2012}, one can start from a curve constructed 
by the method of \cite{broker2009} and take a random walk in the $2$-isogeny graph to get a random 
supersingular curve.

\begin{table}
	\centering
	\small
	\begin{tabular}{cccccc}
		& & \multicolumn{2}{c}{\bfseries Ordinary} & \multicolumn{2}{c}{\bfseries Supersingular} \\
		\cmidrule[0.7pt](r){3-4} \cmidrule[0.7pt](l){5-6} 
		$\#\F_p$ (bits) & Sage pt-cnt & Alg \ref{alg:ss-strong} & IsoGr & Alg 
		\ref{alg:ss-strong} & IsoGr \\
		\midrule[0.7pt]
		33 & 0.035000 & 0.004 & 0.003 & 0.118 & 0.180 \\
		65 & 0.396000 &  0.007 & 0.006 & 0.566 & 0.595 \\
		129 & 5.185000 &  0.007 & 0.006 & 1.074 & 2.783 \\
		257 & 141.3410 &  0.019 & 0.012 & 5.609 & 12.59 \\
		385 & 1839.875 &  0.030 & 0.024 & 19.87 & 32.47 \\
		513 & 3820.591 &  0.049 & 0.039 & 35.87 & 67.66 \\
		641 & 32393.92 &  0.074 & 0.055 & 102.1 & 140.2 \\
		769 & 49644.34 &  0.135 & 0.084 & 157.2 & 256.4 \\
		897 & 96446.71 &  0.162 & 0.110 & 224.2 & 384.8 \\
		1025 & 169138.5 &  0.219 & 0.134 & 322.1 & 556.9 \\		
		\midrule[0.7pt]
	\end{tabular}
	\caption{Experiments (times are in seconds)}
	\label{table:exper}
\end{table}

The average times for each set of curves are listed in columns ``Ordinary" and ``Supersingular". 
Columns ``Alg. \ref{alg:ss-strong}" and ``IsoGr" refer to Algorithm \ref{alg:ss-strong} of Section 
\ref{sec:pit-h-order} and the one in \cite{sutherland2012}, respectively. As the timings suggest, 
complexities of the two algorithms differ only by a constant factor, which confirms the theory. 
Also detecting ordinary curves is substantially faster on average than detecting supersingular 
curves, which again confirms the complexities claimed in Proposition \ref{prop:ss-strong}. 

The second column shows timings for the point counting algorithm in Sage 7.5.1. The 
supersingularity test in Sage is done using a call to the point counting subroutine. Since Sage 
performs naive point counting over the extension $\F_q$, we have used the more efficient underlying 
subroutine \verb|_pari_().ellsea()| from PARI \cite{Pari}. 

\paragraph{Acknowledgment}
The author would like to thank Felipe Voloch for his valuable feedback on Section 
\ref{sec:pit-h-order}, and Luca De Feo for helpful comments. This work was partially supported by 
NSERC, CryptoWorks21, and Public Works and Government Services Canada.


\bibliographystyle{plain}
\bibliography{references}

\end{document}